\documentclass[12pt]{article}

\usepackage[a4paper, total={6in, 9.5in}]{geometry}
\usepackage{cite}
\usepackage{float}
\usepackage{amsmath,amssymb,amsfonts}
\usepackage{amsthm}
\usepackage{algorithmic}
\usepackage{graphicx}
\usepackage{textcomp}
\usepackage{xcolor}
\def\BibTeX{{\rm B\kern-.05em{\sc i\kern-.025em b}\kern-.08em
    T\kern-.1667em\lower.7ex\hbox{E}\kern-.125emX}}

\usepackage{hyperref}
\graphicspath{{./figures/}}
\usepackage{subfig}
\usepackage{caption}
\usepackage{lipsum}
\usepackage[vlined,ruled]{algorithm2e}
\usepackage{tikz}
\usetikzlibrary{arrows.meta}
\usetikzlibrary{shapes.geometric,plotmarks,backgrounds,fit,calc,circuits.ee.IEC}
\usetikzlibrary{decorations.pathreplacing}
\usepackage{stmaryrd}

\newtheorem{theorem}{Theorem}
\newtheorem{corollary}{Corollary}[theorem]
\newtheorem{lemma}{Lemma}
\newtheorem{definition}{Definition}

\DeclareMathOperator{\im}{im}

\begin{document}

\title{Layered Decoding of Quantum LDPC Codes }

\author{
Julien Du Crest\footnote{Univ. Grenoble Alpes, Grenoble INP, LIG, F-38000 Grenoble, France (julien.du-crest@univ-grenoble-alpes.fr)}, Francisco Garcia-Herrero\footnote{Department of Computer Architecture and Automatics (DACYA), Complutense University of Madrid, Spain (francg18@ucm.es) }, Mehdi Mhalla\footnote{ Univ. Grenoble Alpes, CNRS, Grenoble INP, LIG, F-38000 Grenoble, France (mehdi.mhalla@univ-grenoble-alpes.fr)}, \\ Valentin Savin \footnote{ Univ. Grenoble Alpes, CEA-L\'eti, F-38054 Grenoble, France (valentin.savin@cea.fr)} and Javier Valls \footnote{Instituto de Telecomunicaciones y Aplicaciones Multimedia, Universitat Politecnica de Valencia, 46022 Valencia, Spain (jvalls@upv.es)}\\
} 

\maketitle


\begin{abstract}
We address the problem of doing message passing based decoding of quantum LDPC codes under hardware latency limitations. We propose a novel way to do layered decoding that suits quantum constraints and outperforms flooded scheduling, the usual scheduling on parallel architectures. A generic construction is given to construct layers of hypergraph product codes. In the process, we introduce two new notions, \textit{$t$-covering layers} which is a generalization of the usual layer decomposition, and a modified scheduling called \textit{random order scheduling}. Numerical simulations show that the random ordering is of independent interest as it helps relieve the high error floor typical of message passing decoders on quantum codes for both layered and serial decoding without the need of post-processing. 
\end{abstract}

\section{Introduction}
A lot of work has been done in order to improve the decoding of quantum low-density parity-check (qLDPC) codes using message-passing (MP) decoders. 
Most of these works rely on the use of post-processing techniques \cite{panteleev2021degenerate,du2022stabilizer,Raveendran2021trappingsetsof}, whose feasibility is still to be demonstrated on actual hardware, due to the stringent latency, power and scalability requirements of the quantum system. A key attribute of MP decoding is the underlying {\em scheduling}, indicating the order in which variable and check node messages are updated. This has been subject to extensive research in the classical LDPC decoding literature, and it has been shown that the MP scheduling may significantly impact the convergence speed \cite{zhang2007iterative}, the decoding performance (\emph{e.g.}, in case of adaptive scheduling strategies \cite{mao2001decoding, savin2007iterative, vila2010ldpc}), or the performance (\emph{e.g.}, latency, area, power-consumption) of the hardware design \cite{boutillon2014hardware}.  The vast majority of hardware designs are based on partly-parallel architectures, implementing a layered decoding scheduling, which can be considered as a {\em de facto} standard solution, able to provide relevant complexity and performance advantages in most applications \cite{boutillon2014hardware}. 

For qLDPC codes, the MP decoding performance may depend even more on the underlying scheduling, which can be most likely attributed to the code degeneracy~\cite{du2022stabilizer}. Moreover, some post-processing techniques may be highly dependent on the MP decoding scheduling. For instance, the order statistics decoding post-processing has been shown to provide very good performance when a layered scheduling is used, but its performance may be drastically degraded using a flooded (\emph{i.e.}, fully parallel) scheduling~\cite{du2022stabilizer}.

To design an {\em efficient} partly parallel architecture implementing a layered scheduling, one needs a layer decomposition of the parity check matrix. For qLDPC codes this may be tricky, as they do not have an innate decomposition into horizontal layers (as for instance in the case of classical quasi-cyclic LDPC codes).  To ensure a high degree of parallelism, it is also desirable to have a decomposition into a minimal number of layers. In this paper, we first give a generic construction of a minimal layer decomposition for hypergraph product codes. Moreover, in an attempt to start bridging the gap between hardware limitations and state of the art MP decoders, we propose two new tools to implement layered decoding of qLDPC codes. The first is a generalization of the notion of layer decomposition, consisting of a family of \textbf{$t$-covering layers}, which can be seen as a layer decomposition of $t$ decoding iterations, and is aimed at increasing the parallelism degree of the layered architecture. The second is a new scheduling called \textbf{random order scheduling}, and is shown to significantly improve the  decoding performance. Our numerical simulations provide evidence that both could be used in the future to meet hardware needs as they offer a good compromise of speed and performance.

\section{Preliminaries}
\subsection{Quantum Codes}
Calderbank-Shor-Steane (CSS) codes are defined by two classical $(m_X/m_Z,n)$-parity check matrices $H_X,H_Z$ satisfying $H_X H_Z^\perp = 0$. The dimension of the quantum code is $n-\text{rank}(H_X)-\text{rank}(H_Z)$ and its minimum distance $d = \{\min |v|, v \in \ker H_X \backslash \im H_Z^t \cup \ker H_Z \backslash \im H_X^t\}$. 
One such class of CSS codes are the hypergraph product codes (HPC), which given two classical codes $A$ and $B$ give the quantum code $H_X=[A\otimes I, I\otimes B^t]$ and $H_Z=[I\otimes B, A^t \otimes I]$ (see \cite{tillich2013quantum} for the construction and parameters of the code).

In the following, we will only focus on decoding $Z$ errors using $H_X$. All proofs are easily adapted to correcting $X$ errors. 
\subsection{MP decoding}
MP decoders work by exchanging soft information between check and variable nodes on the Tanner graph representation of the code, trying to converge to a hard decision on the variable nodes that satisfy the syndrome. One crucial factor is to decide in which order messages are exchanged and soft information updated. There are three main decoding scheduling used classically: flooded, in which messages are exchanged simultaneously and soft information updated in parallel, serial, where the graph is updated sequentially going through all the checks one by one, and layered, which lies in between, taking advantage of checks that have a disjoint support to update them in parallel, essentially doing a speed up of serial scheduling at no cost. For more details on classical message passing, refer to \cite{savin2014ldpc}.

\subsection{Layered Scheduling}
To avoid memory conflicts in a partly parallel architecture, implementing a layered scheduling, 
the same memory slot should not be read/written to by two different processing units at the same time. This motivates the following definitions. 
\begin{definition}
A \textbf{layer} is a collection of check-nodes such that any two check-nodes have no neighbouring variable-node in common.
\end{definition}
\begin{definition}
A \textbf{layer decomposition} $L_0\sqcup \dots \sqcup L_{k-1}$ is a partition of the set of check-nodes into $k$ layers.
\end{definition}
\begin{definition}
A decomposition is said to be \textbf{minimal} if it is impossible to find a decomposition in less layers.
\end{definition}
A simple density argument is enough to state the following fundamental inequality:
\begin{lemma}
Any $k$ layers decomposition of a $(\_,\delta)$-regular\footnote{A matrix is said to be $(c,d)$-regular if every column is of weight $c$ and every row of weight $d$.} $(m,n)$-parity-check matrix satisfies $k \geq \frac{\delta m}{n}$.
\end{lemma}
\begin{definition}
A decomposition is \textbf{$\gamma$-balanced} if  
$$ \frac{|L_i|}{|L_j|} \leq \gamma,\quad \forall L_i,L_j $$ 
A decomposition will be said to be \textbf{balanced} if it is 1-balanced. Balanced decompositions ensure an efficient use of hardware resources (check-node processing units).
\end{definition}

\section{Hardware Requirements}\label{secHarReq}
In contrast to classical LDPC decoders, which prioritize optimizing throughput, qLDPC ones must satisfy highly constrained values of latency to avoid the backlog problem \cite{holmes2020nisqp}, which would lead to an exponential slowdown of the quantum processor making the QEC implementation impractical.

\begin{table*}[ht!]
    \centering
    \caption{Latency approximation for the different architectures}
    \label{tab:latency}
    \begin{tabular}{c|c|c}
        Parallel & Serial & Layered \\
        \hline
        $T_{min}^{(P)} \times 2 \times it_{max}$ & $T_{min}^{(S)} \times (it_{max} /2) \times m$  &      $T_{min}^{(L)} \times 2 \times (it_{max} /2) \times  k $  \\
 
    \end{tabular}
\end{table*}

To illustrate the behavior of the decoder, the B1 code from~\cite{panteleev2021degenerate} is taken as an example (defined in Section \ref{secGenCon}). An MP decoder for the B1 code can achieve with this architecture\footnote{All figures reported here come from our implementation of an either fully parallel \cite{Valls2021syndrome} or serial min-sum decoder architecture, with exchanged messages quantized on 6 bits, on a Xilinx FPGA xcv095 board.}, a clock period between 8 and 10ns which derives a latency between 480ns and 600ns at 30 iterations, that is close to the most constrained technology. 
Taking this into account, the clock period usually can be reduced to 70\% and 80\% of the clock period obtained with the parallel version. The question is that with this schedule and the derived architecture only one check node is updated in a clock cycle, because of the sequential update of the messages. Due to this, at least $m$ clock-cycles are required \footnote{Assuming that due to the reduced complexity of the units both the check node and the connected variable nodes can be updated in parallel.} to complete just one iteration of the MP algorithm. 
Following the example of code B1, the clock period should be between 5.6 and 7ns, but the total latency at 30 iterations would be between 74.26$\mu$s and 92.82$\mu$s. Even assuming a reduction of the number of iterations to get similar performance to the flooded schedule, the range of latencies would be out of the time budget of supercomputing qubits and transmons. To meet the timing requirements, the clock period should be equal to  $\frac{5.6}{m/2}$=0.025ns, which is a maximum clock frequency of 40GHz. This frequency cannot be achieved by any FPGA or ASIC, and on the other hand, it would require a large power consumption that will cause another problem with the power budget and the refrigeration system \cite{Ueno_2021QECOOL}.
With the previous examples, it is easy to conclude that the implementation of serial scheduling, even if it has a better performance than the flooded one, it is not a realistic solution when it comes to implementation.
A trade-off solution between both serial and flooded may be the layered schedule. If the number of layers is small enough, the number of clock cycles per iteration will not grow too much and the number of iterations will be usually reduced by two. Going back to the B1 code example, assuming that in the worst case, the clock period will be similar to the parallel architecture, and with a distribution of 3.5 layers\footnote{See Section \ref{secGenCon} for the formal definition of fractional layer number.} the total latency for 30 iterations will be between 8x3.5x2x30=1.68$\mu$s and 10x3.5x2x30=2.10$\mu$s; and between 8x3.5x2x15=840ns and 10x3.5x2x15=1.05$\mu$s with 15 iterations, which is fairly close to the constraints of superconducting qubits and meets the requirements of other technologies. 

As we will see in the following sections, the layered schedule will also benefit from some non-negligible performance improvements, apart from the reduction in the number of iterations, compared to the flooded schedule.

In table \ref{tab:latency}, we can find a summary of the total latency for different architectures, where $T_{min}^{(P)}$, $T_{min}^{(S)}$ and $T_{min}^{(L)}$ are the minimum clock periods achievable by the parallel, serial and layered architectures respectively, and $it_{max}$ is the maximum number of iterations configured in the parallel decoder. Note that we assume that the number of iterations of serial and layered is usually half the number of iterations of the parallel architecture~\cite{zhang2007iterative}.

\section{Generic Constructions}\label{secGenCon}

\subsection{Layered Construction for Hypergraph Product Codes}
Consider an hypergraph product code defined by two matrices $A$ and $B$, such that $H_X=[A\otimes I, I\otimes B^t]$ and $H_Z=[I\otimes B, A^t \otimes I]$ \cite{tillich2013quantum}. There exist a layer construction from a layer decomposition of $A,B, A^t, B^t$.
\begin{theorem}
Given minimal decompositions $A=A_0\sqcup\dots \sqcup A_{k_A-1},B = B^t_0\sqcup \dots \sqcup B^t_{k_{B^t}-1}$, one can construct a minimal decomposition of $H_X$ in $k=\max(k_A,k_{B^t})$ layers.

\end{theorem}

A similar theorem can be stated for $A^t, B$ and $H_Z$ with the same proof techniques.
\paragraph{Construction}
If $k_A \neq k_B$, without loss of generality, suppose $k_A < k_B$.  The first step is to add empty layers to A so that $k'_A = k'_B = k$. That is let $A=A_0\sqcup \dots \sqcup A_{k_A-1} \sqcup A_{k_A} \dots \sqcup A_k$ where $A_{k_A} = \dots = A_k = \emptyset$. 
Let's label each row of $H_X$ as  
$$a\varoast b := [a\otimes e_b, e_a \otimes b], \quad
a \in \text{rows}(A), b \in \text{rows}(B^t) $$
In the following we will denote by \textit{left} the sub-matrix $[A\otimes I]$ and \textit{right} the sub-matrix $[I\otimes B^t]$.
Create layers $L_0\dots L_{k-1}$ such that 
\begin{equation}
(a\varoast b) \in L_i \Leftrightarrow \exists j \quad a \in A_j, b \in B^t_{j+i \mod k}
\end{equation}
%

By definition, all checks belong to some layer, we now have to check that any two checks in a given layer have disjoint variable nodes support.
Suppose that two checks $a\varoast b, a' \varoast b'$ belong to $L_i$. 
 Case A : $a=a'$. They do not touch on the left thanks to the tensor product with the identity. Furthermore, it means that $b\neq b'$ but then both belong to $B^t_{l}$ for some $l$, so they have disjoint support on the right.
Case B,C : $a\neq a'$. They have disjoint support on the right because of the tensor product with the identity. To show that they do not intersect on the left, there are two cases :
If $a$ and $a'$ belong to some $A_l$ (case B), then by definition they have disjoint support on the left. 
If $a$ and $a'$ belong respectively to $A_l, A_{l'}$ with $l\neq l'$ (case C), then it means that $b \in B^t_{l+i \mod k}, b'\in B^t_{l'+i \mod k}$, two distinct classes. Hence even though $a$ and $a'$ might share variable nodes in $A$, they do not intersect in the tensored version $A\otimes I$.
Fig. \ref{fig:smallexpl} depicts a simple example of the 3 cases.\\
\setcounter{MaxMatrixCols}{15}
\begin{figure*}[!ht]
\centering
\begin{tabular}{|ccc|ccc|cc|cc|cc|}
\hline
\underline{1} &   &   & \underline{1} &   &   & \underline{1} & \underline{1} &   &   &   &  \\
  & \fbox{1} &   &   & \fbox{1} &   & \fbox{1} &   &   &   &   &  \\
  &   & \fbox{1} &   &   & \fbox{1} &   & \fbox{1} &   &   &   &  \\
\hline
\raisebox{.5pt}{\textcircled{\raisebox{-.9pt} {1}}} &   &   &   &   &   &   &   & \raisebox{.5pt}{\textcircled{\raisebox{-.9pt} {1}}} & \raisebox{.5pt}{\textcircled{\raisebox{-.9pt} {1}}} &   &  \\
  & \underline{1} &   &   &   &   &   &   & \underline{1} &   &   &   \\
  &   & 1 &   &   &   &   &   &   & 1 &   &   \\
\hline
  &   &   & \raisebox{.5pt}{\textcircled{\raisebox{-.9pt} {1}}} &   &   &   &   &   &   & \raisebox{.5pt}{\textcircled{\raisebox{-.9pt} {1}}} & \raisebox{.5pt}{\textcircled{\raisebox{-.9pt} {1}}} \\
  &   &   &   & 1 &   &   &   &   &   & 1 &   \\
  &   &   &   &   & 1 &   &   &   &   &  & 1  \\
\hline
\end{tabular}
    \hspace{2em}
\begin{tabular}{l}
Case \fbox{A} : $a = a' \implies b \neq b', \exists l,\quad b,b' \in B^t_l$ \\
\\
Case \raisebox{.5pt}{\textcircled{\raisebox{-.9pt} {B}}} : $a \neq a',\quad a,a' \in A_l$ \\
\\
Case \underline{C} : $a \neq a', a \in A_l, a'\in A_{l'}$ \\  
\end{tabular}
 
\caption{Small visualization example of proof cases where $A=B^t$}
\label{fig:smallexpl}
\end{figure*}

\paragraph{Minimality} The proof is by contradiction. Assume that there is a decomposition in less than $k_{B^t}$ layers, then one could recover a decomposition for $B^t$ in less than $k_{B^t}$ from a restriction to the $\{a \varoast b,\quad \forall b\}$ positions for any given $a$. Any decomposition in less than $k_A$ layers would similarly give a decomposition for $A$ from the restriction of $H_X$ to any $\{a\varoast b, \quad \forall a\}$ for a given $b$. Hence the decomposition in $\max(k_A,k_B^t)$ is minimal for $H_X$.

Note that the construction is not unique, for example, Equation 1 can be replaced by the following equation where $\sigma$ is any $k$-permutation, although this is still not the most generic formula:
\begin{equation}
(a\varoast b) \in L_i \Leftrightarrow \exists j,\quad a \in A_{\sigma(j)}, b \in B^t_{j+i \mod k}
\end{equation}

\begin{theorem}
Given $k$-layerings for $A$ and $B^t$, respectively $\alpha$ and $\beta$-balanced. Then $H_X$ is $\gamma$-balanced with :\\
\begin{tabular}{ccl}
   (i) & $\gamma < min(\alpha,\beta)$ & if $\alpha,\beta > 1$   \\
    (ii) &  $ \gamma = 1 \qquad$ & otherwise. \\
\end{tabular}

\end{theorem}

\begin{proof}
(i)
Let $a_0...a_{k-1}$ be the sizes of layers $A_0...A_{k-1}$, and $b_0...b_{k-1}$ the sizes of the layers $B^t_0...B^t_{k-1}$. Then each layer $L_l$ of $H_X$ has size $\sum a_i b_{i+l}$. We also suppose layers of $A$ and $B^t$ are ordered from biggest to smallest hence $a_0=\alpha a_{k-1}$ and $b_0=\beta b_{k-1}$.\\
The layer $L_0$ of size $\sum a_i b_i$ is the biggest layer, a classical proof of that is by contradiction, using the fact that $\forall a \geq c, b\geq d,\quad  ab+cd \geq ad+bc $.
The ratio between any other layer $L_j$ and $L_0$ is smaller than $\beta$ since $\beta\sum a_i b_{i+j} > \sum a_i b_0 >  \sum a_i b_i$ using the fact that $\forall b_j,\quad  \beta b_j \geq b_0$ ( and similarly for $\alpha$ ).\\

(ii)
Suppose $A$ is perfectly balanced. In that case, for any $b \in B^t$, it will appear the same 
number of times in each layer $L_i$ since the $a \varoast b, \forall a \in A$ will be equally balanced in the layers. Hence the code will be balanced. The same argument holds if $B^t$ is perfectly balanced.

\end{proof}

\begin{corollary}
Given a $k_A$-layering or $A$, and a $k_B> k_A$ layering for $B^t$ $\beta$-balanced. Then $H_X$ is $\gamma$-balanced with :\\
$$\gamma \leq \beta $$
\end{corollary}
\begin{proof}
    Same as above, considering a $k_B$ layering for $A$ by adding empty layers. This new layering is $\infty$-balanced.
\end{proof}
\subsection{Random Ordering}

We introduce a decoding technique called random ordering. This technique consists of applying a random order on the layers' application at each decoding step. This is also generalized to serial decoding by considering that each check belongs to its own layer (i.e. $k=m$). This seemingly anodyne step helps to alleviate the error floor quite dramatically. In addition, further simulations showed us that one does not even have to use a ``good" pseudo-random generator to generate the permutation, and this can be done with virtually no cost using a simple congruent generator, a solution that is hardware friendly.\\





\subsection{$t$-Covering of Layers}
For many codes, the theoretical bound on $k$ given by a density argument is not tight. However, since for the quantum codes the number of layers is fixed due to latency constraints, it is important to stay as close as possible to the theoretical bound. We introduce a generalization of the layer decomposition called a \textbf{$t$-covering of ($k$) layers}. We drop the requirement that the layers should be disjoint, and only require that their union taken with multiplicities should cover each check exactly $t$ times. 
In the following, the parameters of a $t$-covers will be specified as $(t,k,\gamma)$, giving the cover parameters and the balance of the layers.
Note that when using $t$-covers, the usual term of ``iteration" becomes ambiguous because it might be the case that the decoder stops while all the checks have not been seen the same number of times. Since by pipelining the process, the syndrome satisfaction could be checked after each layer application adding very low latency, in the following we will often refer to the number of iterated layers (but always specify it when we do so).
To quickly compare a $t$-cover with another or with a layer decomposition, it is useful to introduce the  \textbf{fractional layer number} as $\frac{k}{t}$, intuitively it captures the ``average" number of layers the decoder has to process to see each check once. Finally, by concatenating $t$ times the matrix $H$, it is clear that the density bound of lemma 1 applies to the fractional layer number. As a simple application, for the code B1 given below, we found a (2,7,1)-cover, $\frac{k}{t}=3.5$ . We could also find a (1,4,2) layer decomposition, $\frac{k}{t}=4$, and the density bound gives us $\frac{k}{t} \geq 3$, since no decomposition in $3$ layers is known for the B1 code, our $2$-cover sets a new standard in decoding efficiency.

\section{Applications on Particular Quantum Codes}
\subsection{C2 Code}
The C2 code is a hypergraph product code generated from a single cyclic matrix ($A=B$) of generator polynomial $p(x) = 1 + x^2 + x^5$ and length $l=31$.  
Since this cyclic matrix (and its transpose) accepts a decomposition in $5$-layers, using the technique from theorem 1, we can construct a 5 layer decomposition for the C2 code.
As said earlier about the balancing effect of the procedure, the decomposition used for $A,B,A^t,B^t$ is $(1,5,2)$-cover and it yields a $(1,5,1.1)$-cover for C2. This shows the balancing effect of the procedure, as we go from $\gamma=2$ to $\gamma=1.1$. Here are the layers used for the quasicyclic matrix:
\pagebreak
\begin{table}[!ht]
\centering
\begin{tabular}{c|l}
    $A_0$   & 0 1 7 8 14 15 21 22 \\
    $A_1$   & 2 3 9 10 16 17 23 24\\
    $A_2$   & 4 11 18 25 29\\
    $A_3$   & 5 12 19 26 30\\
    $A_4$   & 6 13 20 27 28\\
\end{tabular}
\end{table}

It should be noted that in order to improve the latency (at the cost of a more complex construction), we were also able to create a $(224,961,1)$-cover of C2, achieving a fractional layer number of 4.29 and giving close numerical results. 

\subsection{B1 Code}

\begin{figure*}[ht!]
\centering
\subfloat[{$B1[[882,24]]$ SP}]{\includegraphics[width=0.5\linewidth]{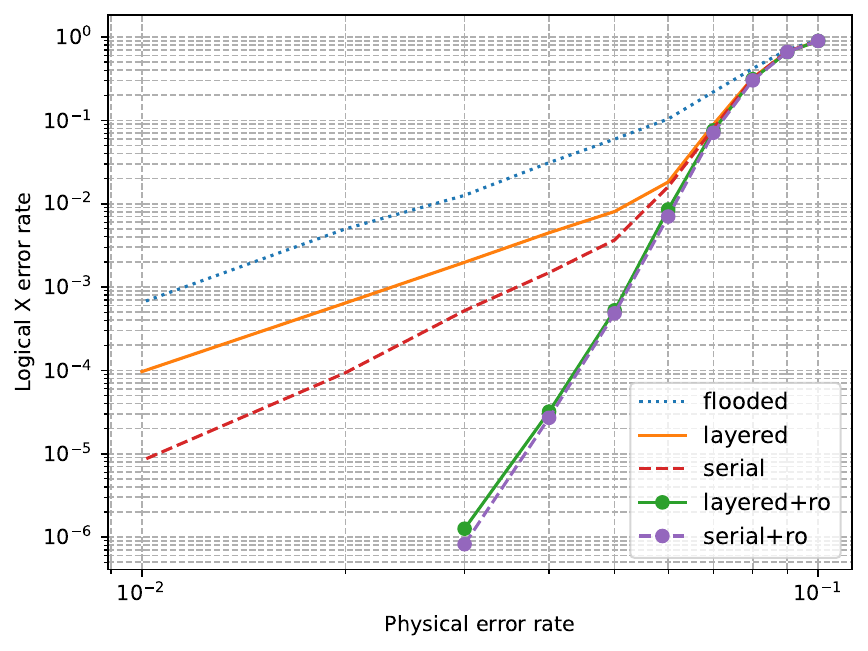}}\hfill%
\subfloat[{$C2[[1922,50,16]]$ SP}]{\includegraphics[width=.5\linewidth]{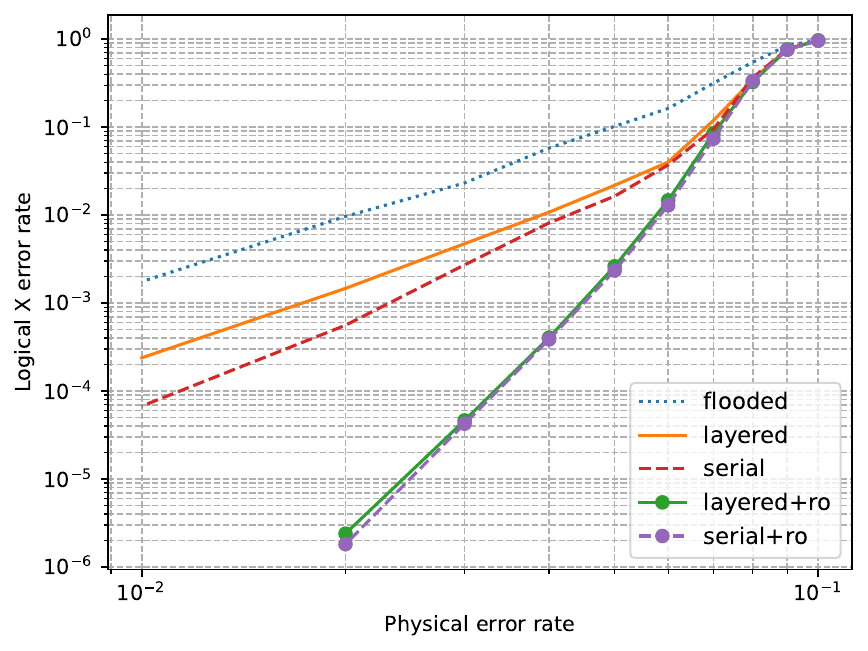}}\\
\subfloat[{$B1[[882,24]]$ NMS}]{\includegraphics[width=0.5\linewidth]{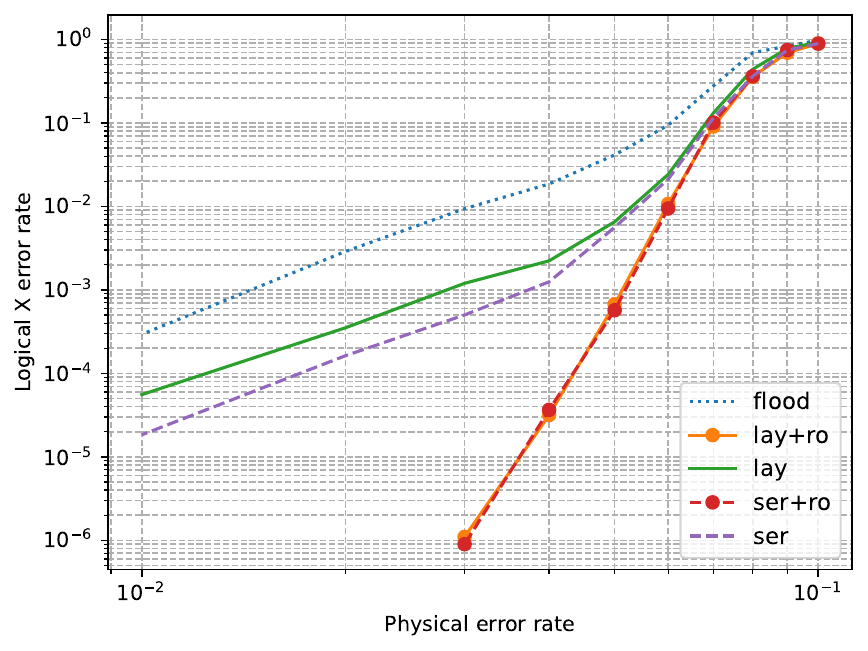}}\hfill%
\subfloat[{$C2[[1922,50,16]]$ NMS}]{\includegraphics[width=.5\linewidth]{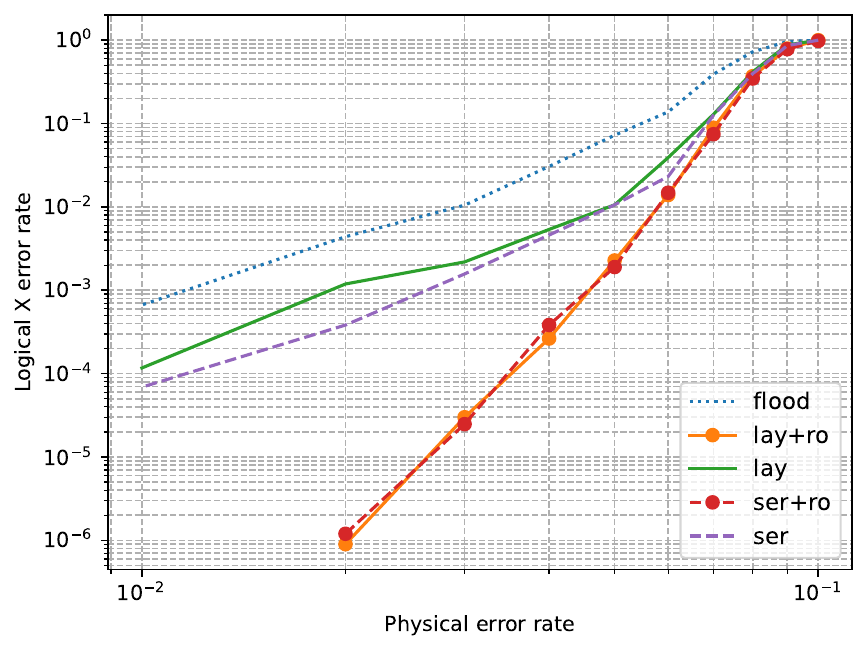}}
\caption{Comparison of different decoders and scheduling on B1 and C2 codes under Z-noise. In the simulations, we use a perturbated NMS, where each check node message is multiplied by a normalization factor uniformly chosen at random in \{0.875,0.9275\} at each iteration.
This perturbation is important to avoid an error-floor degradation.
}
\label{fig:ro}
\end{figure*} 

The B1 code is a Generalized Hypergraph Product code (construction given in \cite{panteleev2021degenerate}, Appendix) :
As such it shares similarities with the Hypergraph Product Codes. Although we do not have a generic decomposition for this family of codes, some ideas from the hypergraph product theorem apply when creating a cover for the B1 code. 
The B1 code accepts a 2-cover in 7 layers, hence giving a fractional layer number of 3.5. The layers are as follows :
$$ L_i = \{ i+7j  \quad \forall j\} \cup \{3+i+7j  \quad \forall j\}$$ 

This 2-cover comes from a decomposition of the quasi-cyclic matrix defined by polynomial $p(x) = 1 + x + x^6$ in 7 layers $S_0\dots S_6$ such that $S_i = \{i+7j \quad \forall j\}$ which is an obvious decomposition (albeit not minimal) given the generating polynomial. Those layers have the property that the union of any two layers $S_i,S_{i+3 \mod 7}$ is still a valid layer which gives the basis for the layers of the code B1. In fact, those layers are extended to the matrices of $H_X,H_Z$ much in the fashion of what is done with the hypergraph product but with a ``twist" as the blocks are quasi-cyclic shifts of identities instead of all identities so one has to be more careful and cannot use the generic formula.



\section{Numerical Results}
Fig. \ref{fig:ro} compares the different decoding techniques proposed on an HGP and Hyperbicycle Codes under Z-noise. We consider Sum-Product (SP) and Normalized Min-Sum (NMS) decoders, with serial, layered and flooded scheduling~\cite{savin2014ldpc}.
When decoding classical codes, using serial scheduling yields a factor two improvement in convergence speed over flooded scheduling. This is not the case in our numerics on quantum codes, as the serial scheduling suffers from a high error floor.
This error floor can be virtually eliminated by using a random ordering scheduling. For the flooded scheduling, the number of iterations used is $i=128$, for serial $i=64$. For the layered scheduling of a $(t,k,\_)$-cover, the layer iteration number $i_{lay} = \lfloor 64\times  \frac{k}{t} \rfloor $. Although we argued before that checking the syndrome after each layer is essentially costless, to do a ``fairer" comparison with serial scheduling under random order scheduling we also tried checking the syndrome only after $\lceil \frac{k}{t}\rceil$ layer iterations. We did not include the numerics as the two curves match almost perfectly, making it a non-issue.

On the B1 code, because it is a $t$-cover and not a layer decomposition, we alter the random ordering scheduling a little bit to boost the performances by requiring that the permutation is not chosen uniformly at random, and must satisfy the additional constraint that two successive layers should not share any check. These additional constraints help the decoder to converge faster as processing the same check twice in a row in different layers would not change its soft information.

\section{Conclusion}
We showed how to implement a layered scheduling for qLDPC codes to meet with the hardware latency limitations. In our numerics, this decoder was more efficient than what could be achieved using similar resources in flooded scheduling which might make it the go to hardware option in the future. We also show that the random order scheduling is a result interesting on its own, as it can be applied to both serial and layered scheduling to alleviate the high error floor of some codes without the need of a post-processing. It should be noted that presently, the best decoders for those codes use some kind of post-processing after message passing, something that was not studied in this paper, as none of the known post-processings can meet the hardware latency considerations. Knowing that our serial scheduling with random ordering already achieves the performances of the Ordered Statistic Decoding (OSD) post-processing on those codes\footnote{See \cite{du2022stabilizer}[Sec 4, Fig. 2], error probability should be multiplied by 2/3 to compare the two since the error model is depolarizing noise there.}.
Finding such hardware friendly post-processing to use with our layered scheduling would be another step in the direction we are aiming for.  

\section*{Acknowledgement}
This work was supported by the QuantERA grant EQUIP (French ANR-22-QUA2-0005-01, and Spain MCIN/AEI/10.13039/501100011033, grant PCI2022-132922),
and the Plan France 2030 (ANR-22-PETQ-0006) and by the European Union “NextGenerationEU/PRTR”.


\bibliographystyle{IEEEtran}
\bibliography{biblio_database}

\end{document}